\documentclass[journal]{IEEEtran}
\IEEEoverridecommandlockouts
\usepackage{cite}
\usepackage{amsmath,amssymb,amsfonts,upgreek}
\usepackage{textcomp}
\usepackage{xcolor}
\usepackage{cite}
\usepackage{url}
\usepackage{mathrsfs}
\usepackage{euscript}
\usepackage{graphicx}
\usepackage{multirow}
\usepackage{changepage}
\usepackage{color}
\usepackage{todonotes}
\usepackage{scalerel}
\usepackage{tikz}
\usetikzlibrary{svg.path}
\usepackage{balance}
\usepackage{svg}
\usetikzlibrary{fit, positioning, shapes.geometric, fit, arrows.meta, calc, backgrounds}
\usepackage{pgfplots}
\usepackage{bm}
\usepackage{algorithm,algorithmic}
\usepackage[caption=false]{subfig}

\floatstyle{ruled}
\newfloat{algorithm}{tbp}{loa}
\providecommand{\algorithmname}{Algorithm}
\floatname{algorithm}{\protect\algorithmname}

\definecolor{darkgreen}{rgb}{0, 0.5, 0} 
\definecolor{lightpurple}{rgb}{0.7, 0.4, 1} 

\definecolor{orcidlogocol}{HTML}{A6CE39}
\tikzset{
orcidlogo/.pic={
\fill[orcidlogocol] svg{M256,128c0,70.7-57.3,128-128,128C57.3,256,0,198.7,0,128C0,57.3,57.3,0,128,0C198.7,0,256,57.3,256,128z};
\fill[white] svg{M86.3,186.2H70.9V79.1h15.4v48.4V186.2z}
svg{M108.9,79.1h41.6c39.6,0,57,28.3,57,53.6c0,27.5-21.5,53.6-56.8,53.6h-41.8V79.1z M124.3,172.4h24.5c34.9,0,42.9-26.5,42.9-39.7c0-21.5-13.7-39.7-43.7-39.7h-23.7V172.4z}
svg{M88.7,56.8c0,5.5-4.5,10.1-10.1,10.1c-5.6,0-10.1-4.6-10.1-10.1c0-5.6,4.5-10.1,10.1-10.1C84.2,46.7,88.7,51.3,88.7,56.8z};
}
}
\newcommand\orcidicon[1]{\href{https://orcid.org/#1}{\mbox{\scalerel*{
\begin{tikzpicture}[yscale=-1,transform shape]
\pic{orcidlogo};
\end{tikzpicture}
}{|}}}}

\usepackage[colorlinks=true,linkcolor=blue,citecolor=blue]{hyperref}
\newtheorem{proof}{Proof}
\newtheorem{proposition}{{Proposition}}

\usepackage[most]{tcolorbox}

\newcommand*{\QED}[1][$\blacksquare$]{%
\leavevmode\unskip\penalty9999 \hbox{}\nobreak\hfill
\quad\hbox{#1}%
}
\begin{document}
\author{
$ 
\text{Jalal Jalali}^{\orcidicon{0000-0002-3609-6775}}\ \IEEEmembership{Member, IEEE}, 
\text{Hina Tabassum}^{\orcidicon{0000-0002-7379-6949}}\ \IEEEmembership{Senior Member, IEEE}
,\
\text{Jeroen Famaey}^{\orcidicon{0000-0002-3587-1354}}\ \IEEEmembership{Senior Member, IEEE}
,\
\newline
\text{Walid Saad}^{\orcidicon{0000-0003-2247-2458}}\ \IEEEmembership{Fellow, IEEE},~
\text{and}\
\text{Murat Uysal}^{\orcidicon{0000-0001-5945-0813}}\ \IEEEmembership{Fellow, IEEE}
\vspace{-5mm}
$.
\thanks{Jalal Jalali
and Jeroen Famaey are with IDLab research group, University of Antwerp - imec, 2000 Antwerp, Belgium (e-mail: \{Jalal.Jalali, 
\!Jeroen.Famaey\}@imec.be).
Jalal Jalali is also with Wireless Communication Research Group, JuliaSpace Inc., Chicago, IL, USA (e-mail: josh@juliaspace.com).
Hina Tabassum is with 
the Department of Electrical Engineering and Computer Science (EECS), 
York University, 
Toronto, ON M3J 1P3, 
Canada
(e-mail: Hina.Tabassum@lassonde.yorku.ca).
Walid Saad is with the Bradley Department of Electrical and Computer Engineering, Virginia Tech, USA (e-mail: walids@vt.edu).
Murat Uysal is with the Engineering Division, New York University Abu Dhabi (NYUAD), Abu Dhabi 129188, UAE (e-mail: murat.uysal@nyu.edu).
}
}
\title{
\LARGE 
Placement, Orientation, and Resource Allocation Optimization for Cell-Free OIRS-aided OWC Network}
\maketitle

\begin{abstract}

The emergence of optical intelligent reflecting surface (OIRS) technologies marks a milestone in optical wireless communication (OWC) systems, enabling enhanced control over light propagation in indoor environments. This capability allows for the customization of channel conditions to achieve specific performance goals. 
This paper presents an enhancement in downlink cell-free OWC networks through the integration of OIRS. The key focus is on fine-tuning crucial parameters, including transmit power, receiver orientations, OIRS elements allocation, and strategic placement. 
In particular, a multi-objective optimization problem (MOOP) aimed at simultaneously improving the network's spectral efficiency (SE) and energy efficiency (EE) while adhering to the network's quality of service (QoS) constraints is formulated. 
The problem is solved by employing the $\epsilon$-constraint method to convert the MOOP into a single-objective optimization problem and solving it through successive convex approximation. Simulation results show the significant impact of OIRS on SE and EE, confirming its effectiveness in improving OWC network performance.
\end{abstract}

\section{Introduction}
\raggedbottom
\indent 
\IEEEPARstart{O}{ptical Wireless communication (OWC)} is gaining recognition as a complement to conventional radio frequency (RF) communications due to its dual functionality, offering both illumination and high-speed data transmission in unregulated license-free spectrum~\cite{9614037}. 
OWC is also renowned for its cost-effectiveness and low energy consumption as it utilizes existing lighting infrastructure~\cite{10462222}. 
In typical downlink OWC networks, the primary roles of transmitters and receivers are fulfilled by light-emitting diodes (LEDs) and photodiodes (PDs), respectively~\cite{9893325}. 
These systems, however, encounter significant line-of-sight (LoS) blockages that can severely impair performance, particularly in indoor environments. 
To address these challenges, \textit{cell-free transmissions} and \textit{optical} intelligent reflecting surface (OIRS) technology offer a potent solution to mitigate blockages by creating alternative wireless propagation paths~\cite{9662064}. 

The OIRSs can be realized through designs based on either mirror arrays or meta-surfaces~\cite{9276478}. It has been shown recently that the mirror array outperforms the metasurface in OWC systems \cite{9500409}. Consequently, a number of recent research works considered optimizing the OWC system in the presence of mirror-based OIRSs. For instance, in \cite{9500409}, the authors maximized secrecy rate while optimizing the orientation of OIRS elements. In \cite{10168927}, the authors derived spectral efficency (SE) and energy efficiency (EE) expressions in the presence of OIRS, but optimization was not considered. Data rate maximization was considered in \cite{9910023, 9543660} by optimizing refractive index and orientations of the OIRS elements, respectively. 
Recently, the potential of OIRS in multiple input multiple output (MIMO) OWC systems has also been demonstrated~\cite{10024150}. 

\textcolor{black}{Although previous works optimized various aspects of OWC systems~\cite{9276478,9500409,9910023,10168927,9543660}, multi-objective \textit{SE and EE maximization} in a \textit{cell-free OIRS-enabled OWC system} presents significant challenges and remains largely unexplored.
The inherent trade-off between SE and EE requires a careful balance: improving SE often demands higher power consumption, which adversely affects EE. Furthermore, the joint optimization of OIRS placement, element assignment, power allocation, and user orientation in a cell-free environment introduces additional computational complexity and non-convexity, necessitating efficient algorithmic solutions.}

\textcolor{black}{In this paper, we develop a framework to jointly maximize the SE and EE of a downlink multi-user OIRS-aided cell-free OWC network. We optimize key system parameters, including OIRS placement, LED-OIRS element assignment, power allocation, and user orientation.
Our proposed approach leverages a multi-objective optimization problem (MOOP) framework. This framework not only improves the total data rate but also EE, offering mathematical rigor and interpretability. 
We propose an iterative, low-complexity algorithm that leverages the $\epsilon$-constraint method to convert the MOOP into a single-objective optimization problem (SOOP) and employs successive convex approximation (SCA) techniques for efficient resolution.
Our findings emphasize the significant impact of OIRS in enhancing SE and EE, illustrating the interplay between the two metrics. Specifically, the proposed method achieves optimal EE at approximately $\rm{3~bits/sec/Hz}$ SE and $\rm{30~dBm}$ optical power, offering a novel and effective approach to designing OIRS-enhanced cell-free OWC networks.}



\textit{Notations:} Scalars, vectors, and matrices are denoted by lowercase italics ($a$), bold lowercase ($\boldsymbol{a}$), and bold uppercase ($\boldsymbol{A}$). Transpose is $\boldsymbol{a}^T$, Hadamard product is $\boldsymbol{A} \circ \boldsymbol{B}$, trace is $\operatorname{tr}(\boldsymbol{A})$, $N$-element ones vector is $\boldsymbol{1}_N$, and $N\times N$ identity matrix is $\boldsymbol{I}_N$. Positive real numbers set is $\mathbb{R}_+$, Euclidean norm is $\|\cdot\|$, and diagonalization is $\rm{diag}(\cdot)$.


\section{OIRS-aided OWC System Description}
\subsection{OWC System Configuration}
We consider the downlink of an OIRS-aided cell-free
OWC system, as depicted in Fig.~\ref{system_model}, where $L$ LEDs serve $K$ 
PD users, with a mirror array-based OIRS with $N$ units enhancing communication. 
We utilize a 3D Cartesian coordinate system, positioning the LEDs at static locations $\boldsymbol{L}_l = [L_{x,l}, L_{y,l}, L_{z,l}]^T \in \mathbb{R}^{3 \times 1}$, the users at $\boldsymbol{u}_k = [u_{x,k}, u_{y,k}, u_{z,k}]^T \in \mathbb{R}^{3 \times 1}$, and the OIRS central location at $\boldsymbol{q} = [q_{x}, q_{y}, q_{z}]^T \in \mathbb{R}^{3 \times 1}$.
We confine the area of interest to four vertical Cartesian planes, $\mathcal{H}_1$ to $\mathcal{H}_4$, where the OIRS could potentially be positioned at: 
\begin{align}
\mathcal{H}_1\!&:\!y_{\min}\!<\!q_y\!<\!y_{\max},
z_{\min}\!<\!q_z\!<\!z_{\max},  q_x\!=\!x_{\min},
\label{H1}\\
\mathcal{H}_2\!&:\!y_{\min}\!<\!q_y\!<\!y_{\max},
z_{\min}\!<\!q_z\!<\!z_{\max},  q_x\!=\!x_{\max},
\label{H2}
\end{align}
\begin{align}
\mathcal{H}_3\!&:\!x_{\min}\!<\!q_x\!<\!x_{\max},
z_{\min}\!<\!q_z\!<\!z_{\max},
q_y\!=\! y_{\min},
\label{H3}\\
\mathcal{H}_4\!&:\!x_{\min}\!<\!q_x\!<\!x_{\max},
z_{\min}\!<\!q_z\!<\!z_{\max},
q_y\!=\! y_{\max}.
\label{H4}
\end{align}
These regions ensure that the OIRS is placed in one of the corner walls of a room-shaped environment.
In this setup, each LED transmits the data of different PDs in a specific time slot, resulting in multi-user interferences (MUI) initiating from different LEDs. 
\textcolor{black}{We use intensity modulation and direct detection (IM/DD), which is a common approach for OWC systems with LEDs. This method is suitable because LEDs are non-coherent light sources, making IM/DD both simple and effective~\cite{9954038}. Our system also employs on-off keying (OOK) for pulse modulation to encode data onto the emitted light. Given the IM/DD assumption, the LED-emitted information symbols are represented by the vector $\boldsymbol{s} = [s_1, \ldots, s_K]^T \in \mathbb{R}^{K \times 1}_+$, with the expectation $E\{|s_k|^2\} = 1, \forall k$.}
These symbols are mapped onto the transmitted signal vector $\boldsymbol{x} = [x_1,\ldots, x_L]^2$. 
The association between $\boldsymbol{s}$ and $\boldsymbol{x}$ is established through 
$\boldsymbol{x}~=~\boldsymbol{A}\boldsymbol{s}$, 
wherein $\boldsymbol{A} = [\boldsymbol{a}_1, \ldots, \boldsymbol{a}_K] \in \mathbb{R}^{L \times K}_+$ 
is a binary matrix with each column vector $\boldsymbol{a}_k = [a_{1,k}, \ldots, a_{L,k}]^T \in \mathbb{R}^{L \times 1}_+$ in 
$\boldsymbol{A}$ adhering to the constraint $\sum_{k=1}^K a_{l,k} = 1, \forall l$.
For our analysis, we assume that the OWC channel state information (CSI) is known at the system IRS controller~\cite{10168927}. 



\vspace{-1mm}
\subsection{LoS/NLoS Channel Model}
\textcolor{black}{In the IM/DD-based OWC system, the LoS channel gain adheres to the Lambertian model~\cite{10190313}.} 
The field of view (FoV) channel gain expression for PD $k$ from LED $l$ is given by:
\begin{equation}
h_{l,k} = 
\frac{C_{\mathrm{PD}}(j+1)}
{2\pi {d}_{l,k}^{2}} 
\cos^{j}(\Psi_{l}) 
{f}_{o} 
\cos(\Omega _{k}) 
{f}_{c},
\forall l,k,
\end{equation}
where $C_{\mathrm{PD}}$ is the PD's physical area, $j$ is the index of Lambertian emission\footnote{The index of Lambertian emission is given by $j=-{\ln 2}/{\ln (\cos \Phi_{1/2})}$, where $\Phi_{1/2}$ is the semi-angle at half power illuminance of the LED~\cite{10024150}.}, and ${d}_{l,k}=\|\boldsymbol{L}_l-\boldsymbol{u}_k\|$ is the distance between LED $l$ and user $k$. 
Angles $\Psi_k$ and $\Omega_k$ are the irradiance and incidence angles for the LoS path from the LED to user $k$. 
Additionally, $f_{o}$ and  $f_{c}$ represent the gains from the optical filter and the optical concentrator. 
For simplicity, the LoS channel gain vector from all LEDs to each PD user is defined as 
$\boldsymbol{h}_k = [h_{1,k},h_{2,k},\ldots,h_{L,k}]^T \in \mathbb{R}^{L \times 1}_+$.

\begin{figure}[ptb]
\centering
\includegraphics[width=.9\linewidth]{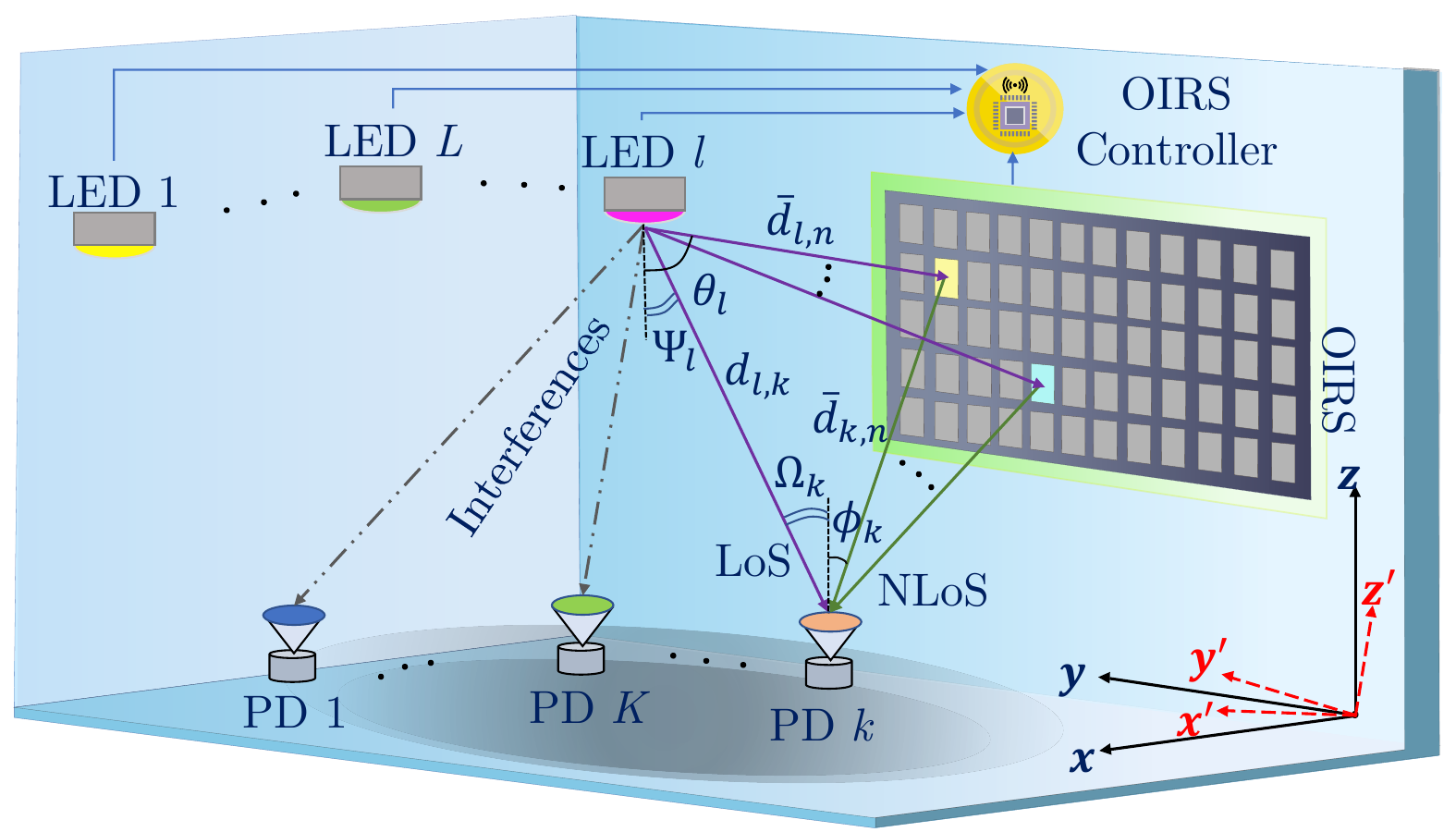}	
\vspace{-4mm}
\caption{Illustration of the OIRS-supported OWC network, where LED $l$ and its reflection are symmetrically positioned across the $x'y'z'$ plane.}
\label{system_model}
\end{figure} 
In an OIRS-assisted OWC system, reflections predominantly involve specular reflections, with minimal diffuse reflections due to relatively low intensity compared to the LoS channel gain~\cite{9681888}.
\textcolor{black}{As an IM/DD system cannot handle phase information, no phase manipulation is involved in OIRS~\cite{10198213}.}
An approximate method for quantifying the OIRS-reflected channel gain is to treat the specular reflection as emanating from a virtual image source~\cite{9662064}. 
In this approach, the OIRS coefficient factor \!$\xi$, additive in nature, is used to account for energy loss during reflection~\cite{9276478}.
Therefore, the gain from the OIRS-reflected channel between the LED and user $k$, facilitated by OIRS unit $n$, can be given as: 
\begin{equation} 
\hspace{-2mm}
{g}_{l,n,k} \!= \!
\xi 
\frac{C_{\mathrm{PD}}(j+1)}
{2\pi({\bar{d}}_{l,n}+{\bar{d}}_{k,n})^{2}} 
\cos^{j}(\theta_{k})
{f}_{o}
\cos(\phi_{k}) 
{f}_{c},
\forall l,n,k,
\end{equation}
where ${\bar{d}}_{l,n}=\|\boldsymbol{L}_l-\boldsymbol{q}\|$ and ${\bar{d}}_{n,k}=\|\boldsymbol{q}-\boldsymbol{u}_k\|$ are the distances from LED $l$ to the OIRS and from the OIRS to PD $k$.
Subsequently, the angles of irradiance and incidence pertaining to the OIRS-reflected path are denoted by $\theta_k$ and $\phi_k$, respectively. 
In addition, for ease of representation, we define 
$\boldsymbol{G}_k = [\boldsymbol{g}_{1,k}, \ldots, \boldsymbol{g}_{L,k}] \in \mathbb{R}^{N \times L}_+$, where 
$\boldsymbol{g}_{l,k} = [g_{l,1,k}, g_{l,2,k}, \ldots, g_{l,N,k}]^T \in \mathbb{R}^{N \times 1}_+$, as the NLoS OIRS-reflected channel matrix of LED $l$ and user $k$.

\vspace{-0mm}
\subsection{Received Signal and SINR}
In the multi-user cell-free OWC framework, each PD is capable of receiving signals from all LEDs.
The LoS signal component arrived at user $k$ can be formulated as:
\begin{equation}
y_k^{\mathrm{LoS}} = 
\delta \sum_{l=1}^{L} 
h_{l,k} P_l x_l = 
\delta \boldsymbol{h}^T_k 
\boldsymbol{P} \boldsymbol{A} \boldsymbol{s}, 
\forall k,
\label{main_y_los}
\end{equation}
where term $h_{l,k} P_l x_l$ represents the useful signal from LED $l$ and corresponds to the received signal for user $k$, whereas the other ($L-1$) terms in the summation constitute the MUI.
In~\eqref{main_y_los}, $P_l$ is the power output of LEDs with 
$\boldsymbol{P}=\mathrm{diag}(P_1,\ldots,P_L)$, and $\delta$ the PD's sensitivity. 
For the NLoS signal, 
we define a binary allocation matrix
$\boldsymbol{B} = [\boldsymbol{b}_1, \ldots, \boldsymbol{b}_L] \in \mathbb{R}^{N \times L}_+$, 
where $\boldsymbol{b}_l = [b_{1,l},\ldots, b_{N,l}]^T \in \mathbb{R}^{N \times 1}_+$ indicates the association of OIRS elements with LEDs. 
If $b_{n,l} = 1$, it implies the element $n$ is dedicated to LED $l$. 
Additionally, the constraint $\sum_{l=1}^{L} b_{n,l} = 1$ is applied to ensure each OIRS element is assigned to only one LED. 
\textcolor{black}{Upon identifying the transmitter, the OIRS unit adjusts its surface based on a reverse look-up table, modifying the orientation and state of its elements. Each OIRS element can be active $1$ or inactive $0$. Active elements reflect incident light from an LED to the intended PD using predetermined angles and positions from the look-up table. This table is based on system geometry and desired reflection paths, ensuring efficient light direction without phase manipulation. The binary vector $\boldsymbol{b}$ controls which elements are active, optimizing communication performance. Due to precise reflection directions and sparse user distribution, unintended reflections are minimal.}
Hence, the interference from OIRS paths is considered insignificant and can be disregarded~\cite{9662064}.
Consequently, the NLoS signal component at user $k$ becomes:
\begin{equation} 
y_{k}^{\mathrm{NLoS}}=
\delta
\boldsymbol{a}_{k}^{T}
\mathrm{diag}
\{\boldsymbol{G}_{k}^{T} \boldsymbol{PB}\}
s_{k},
\forall k,
\end{equation}
where 
$\boldsymbol{g}_{l,k}^{T} \boldsymbol{P}\boldsymbol{b}_l$ is the channel gain from LED $l$ to user $k$, with the sum of all channel gains from LEDs that serve user $k$ being collated.
Given the LoS and NLoS components, user $k$'s received signal becomes:
$
y_k \!= y_k^{\mathrm{LoS}} \!+ y_k^{\mathrm{NLoS}} \!+ n_k, \!
\forall k,
$
where $n_k \!\!\sim \!\!\mathcal{N}(0, \!\sigma^2)$ is the additive white Gaussian noise (AWGN).
Thus, user $k$ signal-to-interference-plus-noise ratio (SINR) can be written as:
\begin{equation} 
\gamma_{k}=
\frac{
\delta^{2}
(
\boldsymbol{h}_k^T \boldsymbol{P}\boldsymbol{a}_{k}+\boldsymbol{a}_{k}^{T}
{\mathrm{diag}}\{\boldsymbol{G}_{k}^{T}\boldsymbol{B}\boldsymbol{P}\}
)^{2}}
{\sigma^{2}+
\delta^{2} \sum_{i=1, i \neq k}^{K}
(\boldsymbol{h}_k^T \boldsymbol{P} \boldsymbol{a}_{i})^{2}}
,\forall k.
\end{equation}
Finally, considering OWC characteristics,  a tight lower bound for user $k$'s data rate can be given as~\cite{9500409}: 
\begin{equation} 
R_{k}
=\frac{1}{2} C
\ln
\left(
{1+\frac{e}{2\pi}\gamma_{k}}\right)
,\forall k,
\end{equation}
where $C \in \mathbb{R}_+$ is the transmitted signal bandwidth.
\textcolor{black}{Efficient optimization of OIRS placement and resource allocation can substantially enhance both spectral and energy efficiency in cell-free OWC networks, mitigating the impact of MUI and LoS blockages and leading to improved reliability and data rates in indoor environments.}

\section{Proposed Algorithm: An SE and EE Tradeoff}
\vspace{2mm}
In this section, we develop an MOOP for the OIRS-assisted OWC system. 
The MOOP 
offers the potential to discover novel approaches for harmonizing and enhancing both SE and EE when adhering to data rate and transmit power constraints.
The EE is expressed as $\eta_{EE} = {R_{\mathrm{tot}}}/{P_{\mathrm{tot}}}$, where $R_{\mathrm{tot}} = \sum_{k=1}^{K} R_k$ is the total data rate, and $P_{\mathrm{tot}} = \mathrm{tr}(\boldsymbol{P}) + p_{\mathrm{cir}}$ is total consumed power with $p_{\mathrm{cir}}$ as a fixed circuit power for the OWC system's operations.
Additionally, the SE is defined as $\eta_{SE}=R_{\mathrm{tot}}/C$. 
The MOOP is thus formulated as follows:
\begin{subequations}
\label{P1_main}
\begin{align}
\text{P}_1: &  \max_{\boldsymbol{P},\boldsymbol{B},\boldsymbol{\Lambda},\boldsymbol{q}}
\text{ \ \ }
\left[\eta_{EE}
(\boldsymbol{P},\boldsymbol{B},\boldsymbol{\Lambda},\boldsymbol{q}),
\eta_{SE}(\boldsymbol{P},\boldsymbol{B},\boldsymbol{\Lambda},\boldsymbol{q})\right]
\label{P1_1}\\
{\rm{s.t.}}: &~
R_k(\boldsymbol{P},\boldsymbol{B},\boldsymbol{\Lambda},\boldsymbol{q})
\geq R_{\min,k}
, ~~ \: \forall k,
\label{p1_c1}\\
& ~\mathrm{tr}(\boldsymbol{P}) \leq P_{\max},
\label{p1_c2}\\
& ~b_{n,l}=\{0,1\},~~~\:\forall n ,l,
\label{p1_c3}\\
& ~\sum_{l=1}^{L} b_{n,l} = 1,
\label{p1_c4}
\\
& ~0 \leq \Omega_{k} \leq 
\Omega_{\mathrm{FoV}},\forall k,
\label{p1_c5}
\\
& ~0 \leq \phi_{k} \leq
\phi_{\mathrm{FoV}},~\forall k, ~\rm{and~}\eqref{H1}-\eqref{H4},
\label{p1_c6}
\end{align}
where $\boldsymbol{\Lambda}=[[\Omega_1,\ldots,\Omega_K]^T,[\phi_1,\ldots,\phi_K]^T]$
is the receiver orientation angles feasibility set.
In \eqref{P1_main}, constraint \eqref{p1_c1} establishes a minimum quality of service (QoS) $R_{\min}$ for each PD, and constraint \eqref{p1_c2} caps the total transmit power of the LED at $P_{\max}=L\Bar{P}_{\mathrm{th}}$, where $\Bar{P}_{\mathrm{th}}$ is the average emitted optical power per LED luminary.
Constraints \eqref{p1_c3} and \eqref{p1_c4} arise from the definition of $\boldsymbol{B}$ while
\eqref{p1_c5} and \eqref{p1_c6} restrict the orientation angles to PDs' FoV in specific ranges. 
Observe that the EE is the quotient of the data rate to power consumption: $R_{\mathrm{tot}} = \eta_{SE}C$. 
From this, it follows 
$\eta_{EE} = {\eta_{SE}C}/{P_{\mathrm{tot}}}$. 
Hence, EE maximization is tantamount to simultaneously maximizing the network's data rate and minimizing its transmit powers. 
Therefore, the MOOP in $\text{P}_1$ can be rewritten as follows:
\end{subequations}
\begin{subequations}
\label{p2_modified}
\begin{align}
\text{P}_2: &  \max_{\boldsymbol{P},\boldsymbol{B},\boldsymbol{\Lambda},\boldsymbol{q}}
\text{ \ \ }
\left[R_{\mathrm{tot}}
(\boldsymbol{P},\boldsymbol{B},\boldsymbol{\Lambda},\boldsymbol{q}),
-P_{\mathrm{tot}}(\boldsymbol{P})
\right]
\\
\rm{s.t.}:&   \ \ \eqref{H1}-\eqref{H4} ~\rm{and~} \eqref{p1_c1}-\eqref{p1_c6}.
\nonumber
\end{align}
\end{subequations}
For resolving the MOOP in \eqref{p2_modified}, we implement the $\epsilon$-method~\cite{Arora}. This method designates one of the objectives as the primary objective and relegates the other objectives to the constraint set. This approach reformulates the transformed MOOP in $\text{P}_2$ into a \!SOOP, which is structured as follows:
\begin{subequations}
\label{p3_main}
\begin{align}
\text{P}_3: & 
\min_{\boldsymbol{P},\boldsymbol{B},\boldsymbol{\Lambda},\boldsymbol{q}}
\text{ \ \ }
\mathrm{tr}(\boldsymbol{P})\\
\rm{s.t.}:&  ~
R_{\mathrm{tot}}
(\boldsymbol{P},\boldsymbol{B},\boldsymbol{\Lambda},\boldsymbol{q})
\leq\epsilon,
\label{p3_c1}
\eqref{H1}\!-\!\eqref{H4} 
~\rm{and~}\! 
\eqref{p1_c1}\!-\!\eqref{p1_c6},
\hspace{-2mm}
\end{align}
where $\epsilon$ denotes the upper bound of the data rate.
To determine the optimal fronts, a mathematical rule for selecting $\epsilon$ within the bounds $R_{\min}\leq\epsilon\leq R_{\max}$ is proposed, where $R_{\max}$ and $R_{\min}$ are the maximal and minimal objective points of total data rate, respectively. 
To fully explore the Pareto fronts, we set $\epsilon\!=\!\alpha\ R_{\max}$, where $\alpha \!\in \!(0,1]$ is a positive scalar. 
Addressing the computational complexity of \eqref{p3_main}, we introduce a four-step iterative method. 
Initially, for a given allocation matrix $\boldsymbol{B}$, a given angle matrix $\boldsymbol{\Lambda}$, and given OIRS placement $\boldsymbol{q}$, the transmit power matrix $\boldsymbol{P}$ is determined. 
After calculating the transmit power, it is used to find the allocation matrix. This matrix then aids in determining the angles, and ultimately, the angle matrix assists in initiating the process of finding OIRS placement.
These steps are repeated iteratively until no additional improvements are achieved. 
Therefore, the iterative procedure for solving $\text{P}_3$ is summarized as:
\end{subequations}
\begin{multline}
\!\!\!\!\!
\underbrace{
\boldsymbol{P}^{(0)}\!\!\rightarrow\!
\boldsymbol{B}^{(0)}\!\!\rightarrow\!
\boldsymbol{\Lambda}^{(0)}\!\!\rightarrow\!
\boldsymbol{q}^{(0)}
}
_{\mathrm{Initialization}}\!\!\rightarrow\!
\ldots
\!\rightarrow\!\!
\underbrace{
\boldsymbol{P}^{(t)}\!\!\rightarrow\!\!
\boldsymbol{B}^{(t)}\!\!\rightarrow\!\!
\boldsymbol{\Lambda}^{(t)}\!\rightarrow\!
\boldsymbol{q}^{(t)}
}
_{\mathrm{Iteration\ t}}
\!\!\rightarrow
\\\!\!\!\!\!\!\!\!
\ldots
\rightarrow
\underbrace{
\boldsymbol{P}^{(\mathrm{opt})}\!\!\rightarrow\!\!
\boldsymbol{B}^{(\mathrm{opt})}\!\!\rightarrow\!\!
\boldsymbol{\Lambda}^{(\mathrm{opt})}\!\rightarrow\!
\boldsymbol{q}^{(\mathrm{opt})}
}_{\mathrm{optimal\ solution}} 
\label{Iter}
\end{multline}
where $t>0$ is the iteration number. 
In \eqref{Iter}, the process begins with an initial feasible solution for the variables ($\boldsymbol{P}^{(0)}, \boldsymbol{B}^{(0)}, \boldsymbol{\Lambda}^{(0)}, \boldsymbol{q}^{(0)}$). 
At each iteration, the approach involves optimizing one variable at a time, keeping the others constant. 
This sequential optimization continues, each time using the previously updated variable while maintaining the others unchanged. The iteration halts when the conditions 
$||\boldsymbol{P}^{(t)}-\boldsymbol{P}^{(t-1)}||\leq\upmu_{1}$,
$||\boldsymbol{B}^{(t)}-\boldsymbol{B}^{(t-1)}||\leq\upmu_{2}$,
$||\boldsymbol{\Lambda}^{(t)}-\boldsymbol{\Lambda}^{(t-1)}||\leq\upmu_{3}$ 
$||\boldsymbol{q}^{(t)}-\boldsymbol{q}^{(t-1)}||\leq\upmu_{4}$,
are met, i.e., convergence, 
where $0<\upmu_{\tau}\ll 1,  \forall \tau=\{1,2,3,4\}$.
\vspace{-1mm}
\subsection{Step 1: LED Emission Power}\label{subA}
We begin with the assumption that the optimal $\boldsymbol{B}$, $\boldsymbol{\Lambda}$, and $\boldsymbol{q}$ are fixed. 
Consequently, the problem in \eqref{p3_main} narrows down to identifying $\boldsymbol{P}$ alone. 
To address the resulting nonconvex problem, we introduce a logarithmic approximation strategy, a form of SCA~\cite{5165179}. 
This technique circumvents the complexities of the nonconcave rate function by establishing a concave lower bound for the rates. 
For this purpose, we define: 
\begin{equation}
v_k\ln(z_{k})+u_k
\leq
\ln(1+z_{k}), \forall k,
\label{app}
\end{equation}
that is a tight (exact) approximation at $z_{k}=\hat{z}_{k}\geq 0$ when the approximation constants are chosen as:
\begin{align}
v_k
=&
\frac{\hat{z}_{k}}
{1+\hat{z}_{k}}, \forall k,
\label{app5}
\\
u_k
=&
\ln(1+\hat{z}_{k})
-\frac{\hat{z}_{k}}{1+\hat{z}_{k}}
\ln(\hat{z}_{k}), \forall k.
\label{app6}
\end{align}
Equations \eqref{app}$-$\eqref{app6} arise from setting equal the slope and function values at each $\hat{z}_{k}$, establishing a distinct correlation between every $\hat{z}_{k}$ and its corresponding pair ${v_{k}, u_{k}}$. 
Utilizing the approximation in \eqref{app} and applying the variable transformation $\tilde{\boldsymbol{P}} = \mathrm{diag}(\ln(P_1),\ldots,\ln(P_L))$, $\text{P}_3$ is rephrased as:
\begin{subequations}
\label{p4_main}
\begin{align}
\text{P}_4: & 
\min_{\tilde{\boldsymbol{P}}}~
\mathrm{tr}
(\exp(\tilde{\boldsymbol{P}})\circ \boldsymbol{I}_{L})
\label{p4_obj}\\
\rm{s.t.}:&
\sum\limits_{k=1}^{K}
v_k\ln(\hat{z}_{k})+u_k
\leq
\epsilon,
\label{p4_c1}\\
& 
v_k\ln(\hat{z}_{k})+u_k
\geq R_{\min},
~\forall k,
\label{p4_c2}\\
& \mathrm{tr}(\exp(\tilde{\boldsymbol{P}})\circ \boldsymbol{I}_{L}) \leq P_{\max},
\label{p4_c3}
\end{align}
\end{subequations}
where 
\begin{align}
\hat{z}_{k}
\!=\!
\frac{e
\delta^{2}
(
\boldsymbol{h}_k^T 
\exp(\tilde{\boldsymbol{P}}) 
\!\circ\! 
\boldsymbol{I}_{L}
\boldsymbol{a}_{k}
\!+\!
\boldsymbol{a}_{k}^{T}
{\mathrm{diag}}\{\boldsymbol{G}_{k}^{T}
\boldsymbol{B}
\exp(\tilde{\boldsymbol{P}})
\!\circ \!\boldsymbol{I}_{L}\}
)^{2}}
{2\pi\sigma^{2}+
2\pi\delta^{2} \sum_{i=1, i \neq k}^{K}
(\boldsymbol{h}_k^T
\exp(\tilde{\boldsymbol{P}})\circ \boldsymbol{I}_{L}
\boldsymbol{a}_{i})^{2}}.
\end{align}
Recognizing the log-sum-exp function is convex, problem \eqref{p4_main} falls into the category of convex problems~\cite{5165179,10459057}. 

\subsection{Step 2: OIRS' Elements Allocation}\label{subB}
Given the transmit power matrix $\boldsymbol{P}$ obtained from the preceding subproblem, along with known $\boldsymbol{\Lambda}$ and $\boldsymbol{q}$ from the previous iteration $t-1$, the next step involves resolving the following feasibility problem to find $\boldsymbol{B}^{}$:
\begin{align}
\label{p5_main}
\text{P}_5: & 
~~\min_{\boldsymbol{B}}
\text{ \ \ }
\boldsymbol{1}^T_N \boldsymbol{1}_L
~~~~~
\rm{s.t.}:
\eqref{p1_c1}, \eqref{p1_c3},\eqref{p1_c4}, \eqref{p3_c1},
\end{align}
where the matrix $\boldsymbol{1}^T_N \boldsymbol{1}_L$ consists entirely of ones, that is $[1]_{n,l}$.   
The integer programming problem \eqref{p5_main}, isomorphic to resource allocation problems, is non-convex and non-deterministic polynomial-time (NP) hard due to the discrete nature of the matrix $\boldsymbol{B}$, resulting in substantial computational complexity. 
To make optimization techniques applicable, a continuous version of the matrix $\boldsymbol{B}$ is employed instead of the binary format. 
Therefore, the constraint in \eqref{p1_c3} is relaxed, enabling a redefinition of \eqref{p5_main} as:
\begin{subequations}
\label{p6_main}
\begin{align}
\text{P}_6: & 
\min_{\boldsymbol{B}}
\text{ \ \ }
\boldsymbol{1}^T_N \boldsymbol{1}_L
\\
\rm{s.t.}:&  ~0 \leq
b_{n,l}
\leq 1
,~~~~~\:~\forall n ,l, 
\label{p6_c1}
\\
&~  0 \leq
b_{n,l} \!-\! b_{n,l}^2 
\leq 1
,\forall n ,l, 
\label{p6_c21}
\rm{and~}
\eqref{p1_c1},\eqref{p1_c4}, \eqref{p3_c1}.
\hspace{-2mm}
\end{align}
\end{subequations}
The relaxation applied to the matrix $\boldsymbol{B}$ transforms the data rate functions in \eqref{p1_c1} and \eqref{p3_c1} into asymptotically concave functions. This change provides a compelling reason to address \eqref{p6_main} using convex optimization techniques, achieving a global optimum~\cite{10456885}. 

\vspace{-3mm}
\subsection{Step 3: PD Orientation Angles}\label{subC}
This section focuses on determining the optimal orientation angles $\boldsymbol{\Lambda}$. 
With $\boldsymbol{P}$, $\boldsymbol{B}$, and $\boldsymbol{q}$ held fixed, we address the nonconvex problem presented in \eqref{p3_main} utilizing the Augmented Lagrangian method. 
This approach merges the Lagrangian function with a quadratic penalty function, as shown in  \eqref{lagran},
where $\kappa$,~$\boldsymbol{\omega }$,~$\boldsymbol{\varpi}$,~and~$\mu$ serve as the Lagrangian multipliers, while $\zeta$ acts as a tunable penalty parameter.
\begin{figure*}
\begin{align}
&  
\mathcal{L}_{\zeta}
\left( 
\boldsymbol{\Lambda}, \boldsymbol{P},
\boldsymbol{q},
\kappa,
\boldsymbol{\varrho},
\boldsymbol{\omega},
\boldsymbol{\varpi}
\right) =
\mathrm{tr}(\boldsymbol{P})
+
\frac{1}{2\zeta}
\Big(
\max
\{0,\kappa+\zeta(\epsilon-R_{\mathrm{tot}})\}^{2}-\kappa^{2}
\Big)
 +
\sum\nolimits_{k=1}^{K}
\Big(
\max
\{0,\varrho_{k}+\zeta(R_{\min}-R_{k})\}^{2}-\varrho_{k}^{2}
\Big)
\nonumber\\
&+
\sum\nolimits_{k=1}^{K}
\Big(
\max
\{0,~\omega_{k}+\zeta(\Omega_{k}-\Omega_{\mathrm{FoV}})\}^{2}-\omega_{k}^{2}
\Big) 
+
\sum\nolimits_{k=1}^{K}
\Big(
\max
\{0,\varpi_{k}+
\zeta
(\phi_{k} - \phi_{\mathrm{FoV}})\}^{2}-\varpi_{k}^{2}
\Big),
\label{lagran}
\end{align}
\hrule
\end{figure*}
These multipliers are given by: 
\begin{align}
\kappa^{(t_{3}+1)}  
& =\max\left\{  0,\kappa^{(t_{3})}+\zeta(
\epsilon-R_{\mathrm{tot}})\right\},
\label{lag_mul_1}\\
\varrho_{k}^{(t_{3}+1)}  
& =\max \left\{0,\varrho_{k}^{(t_{3})}+\zeta \left(R_{\min}-R_{k}\right)\right\},
\label{lag_mul_2}\\
\omega_{k}^{(t_{3}+1)}  
& =\max\left\{0,\omega_{k}^{(t_{3})}+
\zeta(\Omega_{k}-\Omega_{\mathrm{FoV}})\right\},
\label{lag_mul_3}
\end{align}
\begin{align}
\varpi_{k}^{(t_{3}+1)}  
& =\max\left\{  0,\varpi_{k}^{(t_{3})}+
\zeta(\phi_{k} - \phi_{\mathrm{FoV}})\right\}.
\label{lag_mul_4}
\end{align}
The convergence rate of the augmented Lagrangian method, characterized by a constant ratio proportional to $1/\zeta$, exhibits linear behavior when the penalty parameter $\zeta$ is equal to or exceeds a certain positive threshold.

\vspace{-3mm}
\subsection{Step 4: OIRS Placement}\label{subD}
In the last phase, the optimization of the OIRS placement is performed, drawing on the values of $\boldsymbol{P}$, $\boldsymbol{B}$, and $\boldsymbol{\Lambda}$ obtained in the preceding iteration. Given that, the optimization problem for the OIRS placement, $\boldsymbol{q}^{}$, becomes:
\begin{subequations}
\label{p7_main}
\begin{align}
\text{P}_8: & 
\min_{\boldsymbol{q}}
\text{ \ \ }
\boldsymbol{1}_N
\\
\rm{s.t.}:&~
\sum_{k=1}^{K} R_k(\boldsymbol{q}) 
\leq \epsilon,
\label{p7_c1}
\\
&~
R_k(\boldsymbol{q}) 
\geq R_{\min,k},  \: \forall k,~\rm{and~}\eqref{H1}-\eqref{H4}.
\label{p7_c2}
\end{align}
\end{subequations}
Since the objective function in \eqref{p7_main} is convex and the constraints \eqref{p7_c1} and \eqref{p7_c2} demonstrate monotonic concavity in relation to $\boldsymbol{q}$, well-known optimization tools could be used to find an optimal solution for OIRS central position in the multi-user cell-free OWC network. 
Finally, \textbf{Algorithm~\ref{alg_final}} is proposed to determine the LED transmit power matrix, the binary allocation matrix, the PD orientation angles, and the OIRS placement. 
To illustrate the convergence of 
\textcolor{black}{
\begin{proposition}\label{prp_1}
The objective function value of $\text{P}_1$ would be improved via the iterative algorithm.
\end{proposition}
\begin{proof}
See Appendix~\ref{Appen_A}.
\QED
\end{proof}
}




\begin{algorithm}[t]
\caption{\strut Proposed Iterative Algorithm}
\begin{algorithmic}[1]\label{alg_final}
\renewcommand{\algorithmicrequire}{\textbf{Input:}}
\renewcommand{\algorithmicensure}{\textbf{Output:}}
\REQUIRE Set $s=0$, set maximum number of iteration $\mathcal{S}_{\max}$,
initialize
$\boldsymbol{P}=\boldsymbol{P}^{(0)}$,
$\boldsymbol{B}=\boldsymbol{B}^{(0)}$,
$\boldsymbol{\Lambda}=\boldsymbol{\Lambda}^{(0)}$, and
$\boldsymbol{q}=\boldsymbol{q}^{(0)}$.\\    
\STATE \textbf{repeat}\\
\STATE \quad Solve \eqref{p5_main} for given 
$\boldsymbol{B}^{(s-1)}$,
$\boldsymbol{\Lambda}^{(s-1)}$, and
$\boldsymbol{q}^{(s-1)}$ to \\ \quad obtain the optimal solution $\boldsymbol{P}^{(s)}$.
\STATE \quad Solve \eqref{p6_main} for given 
$\boldsymbol{P}^{(s-1)}$,
$\boldsymbol{\Lambda}^{(s-1)}$ and
$\boldsymbol{q}^{(s-1)}$ to \\ \quad obtain the optimal solution $\boldsymbol{B}^{(s)}$.
\STATE \quad Solve 
$\text{P}_7:  \min
\limits_{\boldsymbol{\Lambda}}$ 
$\mathcal{L}_{\zeta }$ in \eqref{lagran}
for given 
$\boldsymbol{P}^{(s-1)}$,
$\boldsymbol{B}^{(s-1)}$, 
\\ \quad $\boldsymbol{q}^{(s-1)}$ to obtain the optimal solution $\boldsymbol{\Lambda}^{(s)}$.
\STATE \quad Solve \eqref{p7_main} for given 
$\boldsymbol{P}^{(s-1)}$,
$\boldsymbol{B}^{(s-1)}$, and
$\boldsymbol{\Lambda}^{(s-1)}$, to \\ \quad obtain the optimal solution $\boldsymbol{q}^{(s)}$.
\STATE   \textbf{until} ${s}=\mathcal{S}_{\max}$ or  convergence
\STATE   \textbf{return} \!\!
$\{
\boldsymbol{P}^{(s)}\!\!,\boldsymbol{B}^{\!(s)}\!\!,\boldsymbol{\Lambda}^{\!\!(s)}\!\!,\boldsymbol{q}^{\!(s)}
\}\!=\!$ 
$\{
\boldsymbol{P}^{(\mathrm{opt})}\!\!,\!\boldsymbol{B}^{(\mathrm{opt})}\!\!,\!\boldsymbol{\Lambda}^{(\mathrm{opt})}\!\!,\!\boldsymbol{q}^{(\mathrm{opt})}
\}$
\end{algorithmic}
\end{algorithm}	

\section{Computational Complexity Analysis}
\textcolor{black}{
The overall computational complexity of the proposed algorithm is determined by solving four subproblems: \( P_4 \), \( P_6 \), \( P_7 \), and \( P_8 \). The complexity of \( P_4 \), based on the SCA, is \( O_1 = \mathcal{O}((K+2)L^3) \). For \( P_6 \), the complexity, aligned with the Interior Point method, is \( O_2 = \mathcal{O}(NK(2NK + K + 2)^2) \). The complexity of \( P_7 \), solved using the Augmented Lagrangian method, is \( O_3 = \mathcal{O}(4K^2) \), and for \( P_8 \), the complexity is \( O_4 = \mathcal{O}(3(K+13)^2) \). Thus, the total complexity of the proposed solution is \( O_{\text{total}} = O_1 + O_2 + O_3 + O_4 = \mathcal{O}((K+2)L^3 + NK(2NK + K + 2)^2 + 4K^2 + 3(K+13)^2) \), indicating a polynomial-time complexity of degree six.
}

\vspace{-1mm}
\section{Simulation Results and Analysis}


In this section, we outline numerical results demonstrating the tradeoff between EE and SE in a room measuring $8 \times 8 \times 3$ meters.
The simulation setup includes key parameters to assess system performance. 
The PD active area, $C_{\mathrm{PD}}$, is $1\mathrm{cm}^2$.
There are four LEDs ($L=4$) evenly distributed on the ceiling at coordinates $(2,2,3)$, $(2,6,3)$, $(6,2,3)$, and $(6,6,3)$ m, with a total modulation bandwidth of $20$ $\mathrm{MHz}$. 
We consider $K = 4$ users and assume noise power is $\sigma^2=-120$ $\mathrm{dBm}$.
The minimum data rate requirement, $R_{\min}=R_{\min,k}$, is set at $0.5$ $\mathrm{bits/sec/Hz}$. 
The circuit power, $p_{\mathrm{cir}}$, is $6.7$ $\mathrm{W}$, and the average emitted optical power of each LED luminary is $\Bar{P}_{\mathrm{th}} = 30$ $\mathrm{dBm}$. 
We use a semi-angle at half power illuminance of $60^{\circ}$,
PD sensitivity of $1$ $\mathrm{A/W,}$ optical filter gain of $1$, a FoV of $80^{\circ}$, and an OIRS reflection coefficient factor $\xi$ of 0.95. 
Finally, the number of OIRS elements is $120$, where each unit has an area of $10 \times 10~\mathrm{cm}^2$. 

\begin{figure}[t]
\vspace{-1mm}
\centering
\begin{tikzpicture}
\begin{axis}[
  width=0.5\textwidth,
  height=0.39\textwidth,
  xlabel={SE [$\mathrm{bits/sec/Hz}$]},
  ylabel={EE [$\mathrm{bits/Joule/Hz}$]},
  xmin=.5, xmax=3.5,
  ymin=.55, ymax=1.3,
  xtick={0.5,1,...,3.5},
  ytick={0.6,0.7,...,1.3},
    legend style={
        at={(0.45,0.695)}, 
        anchor=south east, 
        font=\scriptsize,
        inner sep=0.5mm, 
        legend cell align={left},
        legend columns=1,
        /tikz/column 1/.style={
            column sep=-1pt,
        },
        /tikz/column 2/.style={
            column sep=-1pt,
        },
        /tikz/row 1/.style={
            row sep=-2.5pt,
        },
        /tikz/row 2/.style={
            row sep=-2.5pt, 
        },
        /tikz/row 3/.style={
            row sep=-2.5pt,
        },
        /tikz/row 4/.style={
            row sep=-2.5pt,
        }, 
        /tikz/row 5/.style={
            row sep=-4pt,
        }, 
    },
    tick label style={font=\small},
    xlabel style={font=\footnotesize, yshift=2.5mm},
    ylabel style={font=\footnotesize, yshift=-4.0mm},
  ymajorgrids=true,
  grid=major
]
\addplot[
    black,
    smooth,
    mark=o
] 
coordinates {
(0.5, 0.69)
(0.625, 0.7183)
(0.875, 0.772)
(1.0, 0.8012)
(1.25, 0.8603)
(1.5, 0.9188)
(1.875, 1.0074)
(2, 1.0349)
(2.1875, 1.0829)
(2.5, 1.1583)
(2.625, 1.1845)
(2.75, 1.2107)
(3.0, 1.24)
(3.125, 1.2238)
(3.25, 1.1945)
(3.3125, 1.1483)
(3.35, 1.0436)
(3.375, 0.9388)
};
\addlegendentry{Proposed \textbf{Algorithm~\ref{alg_final}}}

\addplot[smooth,
  color=purple,
  mark=star ,
]
coordinates {
(0.5, 0.677)
(0.625, 0.7)
(0.875, 0.7535)
(1.0, 0.7828)
(1.25, 0.8421)
(1.5, 0.9006)
(1.8125, 0.9751)
(2.0, 1.0155)
(2.375, 1.0958)
(2.5, 1.1227)
(2.625, 1.1457)
(2.8125, 1.1816)
(2.96875, 1.2009)
(3.0, 1.2050)
(3.125, 1.1791)
(3.15625, 1.1413)
(3.2, 1.0436)
(3.225, 0.9388)
(3.230, 0.8798)
};
\addlegendentry{With Fixed $\mathbf{P}$}

\addplot[smooth,
  color=orange,
  mark=diamond,
]
coordinates {
(0.5, 0.615)
(0.5655, 0.635)
(1.0, 0.7207)
(1.1542, 0.7567)
(1.385, 0.8019)
(1.7313, 0.8662)
(2.0198, 0.9114)
(2.3083, 0.9567)
(2.4238, 0.9748)
(2.5392, 0.9929)
(2.77, 1.02)
(2.8854, 1.019)
(3.0008, 0.9868)
(3.0585, 0.9567)
(3.0932, 0.8843)
(3.1163, 0.8119)
};
\addlegendentry{Random PD Orientation}

\addplot[smooth,
  color=blue,
  mark=triangle,
]
coordinates {
(0.5, 0.585)
(1, 0.675)
(1.125, 0.6959)
(1.35, 0.7375)
(1.6875, 0.7916)
(2, 0.8292)
(2.25, 0.8529)
(2.3625, 0.8648)
(2.5, 0.8737)
(2.61, 0.88)
(2.7, 0.8737)
(2.7788, 0.8611)
(2.8125, 0.8105)
(2.835, 0.76)
(2.845, 0.70)
};
\addlegendentry{Random $\mathbf{B}$ Association}

\addplot[
  smooth,
  color=red,
  mark=square,
]
coordinates {
(0.5,0.555)
(1,0.645)
(1.125,0.6659)
(1.35,0.7075)
(1.6875,0.7606)
(2,0.7992)
(2.25,0.8299)
(2.4,0.8437)
(2.5,0.85)
(2.6,0.8437)
(2.6788,0.8311)
(2.7125,0.7805)
(2.735,0.73)
(2.745,0.67)
};
\addlegendentry{No IRS}

\end{axis}
\end{tikzpicture}
\vspace{-6mm}
\caption{EE versus SE for different heights of the LEDs in the OWC network.}
\label{EE_vs_SE_fig}
\end{figure}

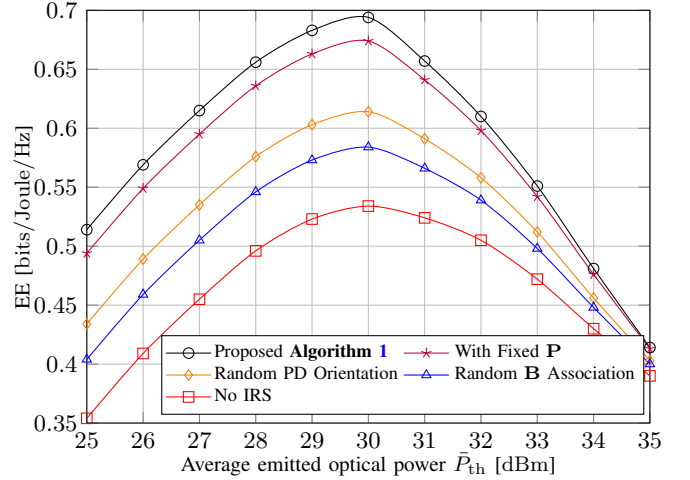
\begin{figure}[t]
\vspace{-1mm}
\centering
\begin{tikzpicture}
\begin{axis}[
  width=0.5\textwidth,
  height=0.39\textwidth,
  xlabel={Average emitted optical power $\Bar{P}_{\mathrm{th}}$ [$\mathrm{dBm}$]},
  ylabel={EE [$\mathrm{bits/Joule/Hz}$]},
  after end axis/.code={
  },
  xmin=25, xmax=35,
  ymin=0.35 , ymax=0.70 ,
  xtick={25,26,...,35},
  ytick={0.35 ,0.40 ,...,0.70 },
    legend style={
        at={(0.99,0.02)}, 
        anchor=south east, 
        font=\scriptsize,
        inner sep=0.5mm, 
        legend cell align={left},
        legend columns=2,
        /tikz/column 1/.style={
            column sep=-1pt,
        },
        /tikz/column 2/.style={
            column sep=0pt,
        },
        /tikz/row 1/.style={
            row sep=-2pt,
        },
        /tikz/row 2/.style={
            row sep=-2pt, 
        },
        /tikz/row 3/.style={
            row sep=-5pt,
        },
    },
    tick label style={font=\small},
    xlabel style={font=\footnotesize, yshift=2.5mm},
    ylabel style={font=\footnotesize, yshift=-4.0mm},
  ymajorgrids=true,
  grid=major
]

\addplot[smooth,
  color=black,
  mark=o,
]
coordinates {
(25, 0.514)
(26, 0.569)
(27, 0.615)
(28, 0.656)
(29, 0.683)
(30, 0.694)
(31, 0.657)
(32, 0.610)
(33, 0.551)
(34, 0.481)
(35, 0.414)
};
\addlegendentry{Proposed \textbf{Algorithm~\ref{alg_final}}}

\addplot[smooth,
  color=purple,
  mark=star ,
]
coordinates {
(25, 0.494)
(26, 0.549)
(27, 0.595)
(28, 0.636)
(29, 0.663)
(30, 0.674)
(31, 0.641)
(32, 0.598)
(33, 0.542)
(34, 0.476)
(35, 0.413)
};
\addlegendentry{With Fixed $\mathbf{P}$}

\addplot[smooth,
  color=orange,
  mark=diamond,
]
coordinates {
(25, 0.434)
(26, 0.489)
(27, 0.535)
(28, 0.576)
(29, 0.603)
(30, 0.614)
(31, 0.591)
(32, 0.558)
(33, 0.512)
(34, 0.456)
(35, 0.403)
};
\addlegendentry{Random PD Orientation}

\addplot[smooth,
  color=blue,
  mark=triangle,
]
coordinates {
(25, 0.404)
(26, 0.459)
(27, 0.505)
(28, 0.546)
(29, 0.573)
(30, 0.584)
(31, 0.566)
(32, 0.539)
(33, 0.498)
(34, 0.448)
(35, 0.400)
};
\addlegendentry{Random $\mathbf{B}$ Association}

\addplot[
  smooth,
  color=red,
  mark=square,
]
coordinates {
    (25, 0.354) 
    (26, 0.409) 
    (27, 0.455) 
    (28, 0.496)
    (29, 0.523) 
    (30, 0.534) 
    (31, 0.524) 
    (32, 0.505)
    (33, 0.472) 
    (34, 0.430) 
    (35, 0.390)
};
\addlegendentry{No IRS}

\end{axis}
\end{tikzpicture}
\vspace{-6mm}
\caption{EE vs. average LET optical transmit power with $\alpha=0.6$.}\label{EE_vs_max_p_fig}
\end{figure}

Fig. \ref{EE_vs_SE_fig} illustrates the interplay between EE and SE with different benchmarks. 
A bell-shaped curve is observed, indicating an initial increase in EE with SE, reaching a peak, followed by a subsequent decrease. 
\textcolor{black}{This trend is most prominent with the proposed \textbf{Algorithm~\ref{alg_final}}, which achieves superior performance by effectively optimizing the trade-off between EE and SE. The results peak at approximately $3$ $\mathrm{bits/sec/Hz}$, corresponding to the optimal operational point for the algorithm with the OIRS central placement optimized at $\boldsymbol{q} = [0, 4.1, 2]^T \in \mathcal{H}_1$.}
The bell-shaped trend correlates with how EE initially increases with the total data rate ($R_{\mathrm{tot}}$), peaks, and then falls due to the rising consumed power ($P_{\mathrm{tot}}$).The incremental power boosts the denominator of $\eta_{EE}$ more significantly than the gains in $R_{\mathrm{tot}}$ in the numerator, leading to a net decrease in EE.  This demonstrates the inherent trade-off in OIRS-aided OWC networks where, after a certain point, higher throughput does not lead to energy savings.
Scenarios with random PD orientation, random $\boldsymbol{B}$ association, and fixed transmit power $\boldsymbol{P}$ highlight the significance of strategic system configuration for efficiency. The marked performance dip without OIRS underscores its crucial role in boosting EE and SE. However, the benefits of using OIRS are limited by diminishing returns, an essential factor in OWC design prioritizing energy conservation and data transmission effectiveness.
\textcolor{black}{Furthermore, it is observed that fixed power allocation has a less significant effect on the overall system performance compared to the orientation and element associations. This is because once the OIRS placement is optimized, power allocation adjusts only the signal strength over the already-determined reflection paths. In contrast, orientation and element associations provide fine-grained control over the reflection angles and user-specific signal alignment, leading to more substantial gains in spectral and energy efficiency by minimizing interference and improving signal quality.}





\vspace{0.5mm}
Fig. \ref{EE_vs_max_p_fig} explores the relationship between EE and the average emitted optical power $\Bar{P}_{\mathrm{th}}$. 
It is observed that the system's EE, as proposed in our algorithm, escalates with an increase in $\Bar{P}_{\mathrm{th}}$. 
Specifically, as $\Bar{P}_{\mathrm{th}}$ rises, there is a notable enhancement in the system's EE, which eventually reaches a saturation point around 
$30$ $\mathrm{dBm}$.  
Clearly, regardless of the chosen benchmark, there exists an optimal emission power level at which EE peaks before it begins to recede. Past this peak, the resource allocator refrains from increasing power output, recognizing that further power boosts do not correspond to equivalent gains in efficiency, potentially due to escalating interference.



\section{Conclusions}

In this paper, we have explored an OIRS-assisted cell-free downlink OWC system. 
We introduced a novel tradeoff between EE and SE through a MOOP that concurrently maximizes different objectives within specific data rates and power constraints. 
Our findings emphasize the critical influence of receiver orientation, OIRS placement, OIRS elements' assignment, and power control on the system's efficiency, particularly in mitigating multi-user interference.
\textcolor{black}{Exploring the significance of artificial intelligence (AI) in solving such problems is a promising direction for future research.}

\textcolor{black}{
\appendices
\vspace{-6mm}
\section{Proof of the Proposition~\ref{prp_1}}\label{Appen_A}
Let us consider 
$\{\boldsymbol{P}^{(j)},\boldsymbol{B}^{(j)},\boldsymbol{\Lambda}^{(j)}, \boldsymbol{q}^{(j+1)}\}$
as the feasible solution set to $\text{P}_8$.
Then, the feasible solution set of $\text{P}_8$ is a feasible solution to $\text{P}_1$ as well. 
Therefore, the following 
$\{\boldsymbol{P}^{(j)},\boldsymbol{B}^{(j)},\boldsymbol{\Lambda}^{(j)}, \boldsymbol{q}^{(j)}\}$
and 
$\{\boldsymbol{P}^{(j+1)},\boldsymbol{B}^{(j+1)},\boldsymbol{\Lambda}^{(j+1)}, \boldsymbol{q}^{(j+1)}\}$
are feasible to $\text{P}_1$  in the $(j)$-th and $(j + 1)$- th iterations, respectively. 
Now, we define 
$f_{\text{P}_1}(\boldsymbol{P}^{(j)},\boldsymbol{B}^{(j)},\boldsymbol{\Lambda}^{(j)}, \boldsymbol{q}^{(j)})$, 
$f_{\text{P}_8}(\boldsymbol{q}^{(j)})$, 
$f_{\text{P}_7}(\boldsymbol{\Lambda}^{(j)})$, 
$f_{\text{P}_6}(\boldsymbol{B}^{(j)})$, 
and 
$f_{\text{P}_4}(\boldsymbol{P}^{(j)})$
as the objective functions of problems $\text{P}_1$, 
$\text{P}_8$, $\text{P}_7$, $\text{P}_6$ and $\text{P}_4$ in the $(j)$-th iteration, respectively. 
Thus, we have:
\begin{align}
&f_{\text{P}_1}(\boldsymbol{P}^{(j+1)},\boldsymbol{B}^{(j+1)},\boldsymbol{\Lambda}^{(j+1)}, \boldsymbol{q}^{(j+1)})
\nonumber\\    
&\overset{(a)}{=}
f_{\text{P}_8}(\boldsymbol{q}^{(j+1)})
\overset{(b)}{\geq}
f_{\text{P}_8}(\boldsymbol{q}^{(j)})
\nonumber\\
&=
f_{\text{P}_1}(\boldsymbol{P}^{(j)},\boldsymbol{B}^{(j)},\boldsymbol{\Lambda}^{(j)}, \boldsymbol{q}^{(j)}),
\end{align}
where $(a)$ follows the fact that the problem $\text{P}_1$ is equivalent to the problem $\text{P}_8$ for optimal $\boldsymbol{P}$, $\boldsymbol{B}$ and $\boldsymbol{\Lambda}$, and 
$(b)$ holds since 
$f_{\text{P}_8}(\boldsymbol{q}^{(j+1)})
{\geq}
f_{\text{P}_8}(\boldsymbol{q}^{(j)})$
according to sub-problem \ref{subD} (optimizing the placement of the OIRS).
Similarly, for a given ${\boldsymbol{P}^{(j)},\boldsymbol{B}^{(j)}, \boldsymbol{q}^{(j)}}$, we have:
\begin{align}
&f_{\text{P}_1}(\boldsymbol{P}^{(j+1)},\boldsymbol{B}^{(j+1)},\boldsymbol{\Lambda}^{(j+1)}, \boldsymbol{q}^{(j+1)})
\nonumber\\    
&\overset{(a)}{=}
f_{\text{P}_7}(\boldsymbol{\Lambda}^{(j+1)})
\overset{(b)}{\geq}
f_{\text{P}_7}(\boldsymbol{\Lambda}^{(j)})
\nonumber\\
&=
f_{\text{P}_1}(\boldsymbol{P}^{(j)},\boldsymbol{B}^{(j)},\boldsymbol{\Lambda}^{(j)}, \boldsymbol{q}^{(j)}),
\end{align}
where $(a)$ follows the fact that the problem $\text{P}_1$ is equivalent to the problem $\text{P}_7$ for optimal $\boldsymbol{P}$ and $\boldsymbol{B}$, and 
$(b)$ holds since 
$f_{\text{P}_7}(\boldsymbol{\Lambda}^{(j+1)})
{\geq}
f_{\text{P}_7}(\boldsymbol{\Lambda}^{(j)})$
according to sub-problem \ref{subC} (optimizing the PD orientation angles).
Similarly, for a given ${\boldsymbol{P}^{(j)},\boldsymbol{\Lambda}^{(j)}, \boldsymbol{q}^{(j)}}$, we have:
\begin{align}
&f_{\text{P}_1}(\boldsymbol{P}^{(j+1)},\boldsymbol{B}^{(j+1)},\boldsymbol{\Lambda}^{(j+1)}, \boldsymbol{q}^{(j+1)})
\nonumber\\    
&\overset{(a)}{=}
f_{\text{P}_6}(\boldsymbol{B}^{(j+1)})
\overset{(b)}{\geq}
f_{\text{P}_6}(\boldsymbol{B}^{(j)})
\nonumber\\
&=
f_{\text{P}_1}({\boldsymbol{P}^{(j)},\boldsymbol{B}^{(j)},\boldsymbol{\Lambda}^{(j)}, \boldsymbol{q}^{(j)}}).
\end{align}
where $(a)$ follows the fact that problem $\text{P}_1$ is equivalent to problem $\text{P}_6$ for optimal $\boldsymbol{P},\boldsymbol{\Lambda}$, and  $\boldsymbol{q}$, and 
$(b)$ holds since 
$f_{\text{P}_6}(\boldsymbol{B}^{(j+1)})
{\geq}
f_{\text{P}_6}(\boldsymbol{B}^{(j)})$
according to sub-problem \ref{subB} (the IRS optimal placement).
Equivalently, for a given $\boldsymbol{B}^{(j)},\boldsymbol{\Lambda}^{(j)}, \boldsymbol{q}^{(j)}$, we have:
\begin{align}
&f_{\text{P}_1}(\boldsymbol{P}^{(j+1)},\boldsymbol{B}^{(j+1)},\boldsymbol{\Lambda}^{(j+1)}, \boldsymbol{q}^{(j+1)})
\nonumber\\    
&\overset{(a)}{=}
f_{\text{P}_4}(\boldsymbol{P}^{(j+1)})
\overset{(b)}{\geq}
f_{\text{P}_4}(\boldsymbol{P}^{(j)})
\nonumber\\
&=
f_{\text{P}_1}({\boldsymbol{P}^{(j)},\boldsymbol{B}^{(j)},\boldsymbol{\Lambda}^{(j)}, \boldsymbol{q}^{(j)}}).
\end{align}
where $(a)$ follows the fact that the problem $\text{P}_1$ is equivalent to the problem $\text{P}_4$ for optimal $\boldsymbol{B},\boldsymbol{\Lambda}$, and  $\boldsymbol{q}$, and 
$(b)$ holds since 
$f_{\text{P}_4}(\boldsymbol{P}^{(j+1)})
{\geq}
f_{\text{P}_4}(\boldsymbol{P}^{(j)})$
according to sub-problem \ref{subA} (the LED emission power control).
From the above four inequalities, we can conclude the following inequality holds: 
\begin{align}
&f_{\text{P}_1}(\boldsymbol{P}^{(j+1)}\!,\boldsymbol{B}^{(j+1)}\!,\boldsymbol{\Lambda}^{(j+1)}\!, \boldsymbol{q}^{(j+1)})
{\geq}
f_{\text{P}_1}({\boldsymbol{P}^{(j)}\!,\boldsymbol{B}^{(j)}\!,\boldsymbol{\Lambda}^{(j)}\!, \boldsymbol{q}^{(j)}}).
\end{align}
Thus, we have shown that the objective function of $\text{P}_1$ is monotonically non-decreasing after each iteration.\QED
}
\vspace{-3mm}
\bibliographystyle{ieeetr}
\bibliography{ref}

\end{document}